\documentclass[11pt]{llncs}
\usepackage{amsmath,amssymb,amsfonts}
\usepackage{gastex}
\usepackage{graphicx}
\usepackage{epic}

\title{Testing for Synchronization}
\author{Mikhail V. Berlinkov\thanks{Supported by the Russian Foundation for Basic Research,
grant 13-01-00852, and by the Presidential Program for young
researchers, grant MK-266.2012.1.}}

\authorrunning{M. V. Berlinkov}

\tocauthor{M. V. Berlinkov (Ekaterinburg, Russia)}

\institute{ Institute of Mathematics and Computer Science\\
              Ural Federal University, 620000 Ekaterinburg, Russia\\
              \email{berlm@mail.ru}}

\DeclareSymbolFont{rsfscript}{OMS}{rsfs}{m}{n}
\DeclareSymbolFontAlphabet{\mathrsfs}{rsfscript}

\newcommand{\sw}{reset word}

\newtheorem{crit}{Criterion}

\begin{document}
\maketitle

\begin{abstract}
We consider the first problem that appears in any application of
synchronizing automata, namely, the problem of deciding whether or
not a given $n$-state $k$-letter automaton is synchronizing. First
we generalize results from~\cite{RandSynch},\cite{On2Problems} for
the case of strongly connected partial automata. Specifically, for
$k>1$ we show that such an automaton is synchronizing with
probability $1-O(\frac{1}{n^{0.5k}})$ and present an algorithm with
linear in $n$ expected time, while the best known algorithm is
quadratic on each instance. This results are interesting due to
their applications in synchronization of finite state information
sources.

After that we consider the synchronization of reachable partial
automata that has application for splicing systems in computational
biology. For this case we prove that the problem of testing a given
automaton for synchronization is NP-complete.
\end{abstract}

\section{Preliminaries}

A \emph{deterministic finite automata} (DFA) $\mathrsfs{A}$ is a
triple $\langle Q,\Sigma,\delta \rangle$ where $Q$ is the state set,
$\Sigma$ is the input alphabet and $\delta: Q \times \Sigma
\rightarrow Q$ is the \emph{transition function}. If $\delta$ is
completely defined on $Q \times \Sigma$ then $\mathrsfs{A}$ is
called \emph{complete}, otherwise $\mathrsfs{A}$ is called
\emph{partial}. The function $\delta$ extends uniquely to a function
$Q\times\Sigma^*\to Q$, where $\Sigma^*$ stands for the free monoid
over $\Sigma$; the latter function is still denoted by $\delta$.
When we have specified a DFA $\mathrsfs{A}=\langle
Q,\Sigma,\delta\rangle$, we can simplify the notation by writing
$S.w$ instead of $\{ \delta(q,w) \mid q \in S \}$ for a subset
$S\subseteq Q$ and a word $w \in \Sigma^*$. In what follows, we
assume $|\Sigma|>1$ because the singleton alphabet case is trivial
for considered problems.

A DFA $\mathrsfs{A} = \langle Q,\Sigma,\delta \rangle$ is called
\emph{synchronizing} if there exists a word $w\in\Sigma^*$ such that
$|Q.w|=1$. Notice that here $w$ is not assumed to be defined at all
states. Each word $w$ with this property is said to be a
\emph{reset} or \emph{synchronizing} word for $\mathrsfs{A}$.

The synchronization of strongly connected partial automata as models
of \emph{$\epsilon$-machines} is one of the central object for
research in the theory of stationary information sources. The
synchronization and state prediction for stationary information
sources has many applications in information theory and dynamical
systems. An $\epsilon$-machine can be defined as a strongly
connected DFA with probability distribution defined on outgoing
arrows for each state (see~\cite{Emach},\cite{Amach} for details).
An $\epsilon$-machine is \emph{exactly synchronizable} or simply
\emph{exact} if the corresponding partial strongly connected
automaton is synchronizing in our terms.

A word $v$ \emph{merges} a pair $\{p,q\}$ if $p.v = q.v$ or $v$ is
defined on exactly one of the states from $\{p,q\}$. The following
analogue of synchronization criterion from~\cite{Ce64} for this case
also has been presented in~\cite{Emach}.
\begin{crit}[Travers and Crutchfield~\cite{Emach}]
\label{e_crit}A strongly connected partial automaton is
synchronizing if and only if for each pair of states $p,q \in Q$
there is a word $v$ which \emph{merges} the pair $\{p,q\}$.
\end{crit}
Given a partial strongly connected DFA $\mathrsfs{A} = \langle
Q,\Sigma,\delta \rangle$, this criterion can be verified by running
\emph{Breadth First Search} (\emph{BFS}) from the set $$\{ \{q,q\}
\mid q \in Q\} \cup \{\{0,q\} \mid q \in Q\}$$ by reverse arrows in
the \emph{square automaton} $\mathrsfs{A}^2 =  \langle Q^{2},
\Sigma, \delta^2 \rangle$, where $Q^{2} = \{ \{p,q\} \mid p, q \in Q
\cup \{0\} \}$, and $\delta^2$ is the natural extension of $\delta$
to $Q^2$ where all undefined transitions are replaced with
transitions to $0$, i.e. for each $x \in \Sigma$
\begin{equation}\delta^2(\{p,q\},x) = \begin{cases}
\delta(\{p,q\},x),\ & p \neq 0, q \neq 0 \\
\{\delta(p,x),0\},\ & q = 0
\end{cases}
\end{equation}
Let us call this algorithm \emph{$IsSynch$}. Since $\mathrsfs{A}^2$
has $|\Sigma|(n+1)^2$ arrows, this algorithm is quadratic in time
and space. Notice that this algorithm is quadratic in $n$ for each
automaton whence it is expected time for \emph{random} automata is
also quadratic. In Section~\ref{sec_partial} we generalize results
presented in~\cite{RandSynch},\cite{On2Problems}. First we show that
a \emph{random} strongly connected partial automaton $\mathrsfs{A} =
\langle Q,\Sigma,\delta \rangle$ is synchronizable with probability
$1-O(\frac{1}{n^{0.5|\Sigma|}})$ and the bound is tight for the
binary alphabet case. As well as in~\cite{On2Problems}, this result
yields an algorithm having linear in $n$ expected time.

If each pair of states in a complete automaton $\mathrsfs{A}$ can be
merged then $\mathrsfs{A}$ is synchronizing. Hence $IsSynch$ can be
used to test a complete automaton for synchronization because the
strong connectivity condition is not used by the algorithm. However,
$IsSynch$ can not be used to test a given partial DFA for
synchronization because the strong-connectivity condition is
essential in Criterion~\ref{e_crit}. Let us consider the case of
\emph{reachable} partial automata in details. Recall that an
automaton is called reachable if one can choose an \emph{initial}
state $q_0$ and a \emph{final} set of states $F$ such that each
state $q \in Q$ is accessible from $q_0$ and co-accessible from $F$,
i.e. there are words $u,v \in \Sigma^*$ such that $q_0.u = q$ and
$q.v \in F$.

This case is of certain interest due to its applications in
dna-computing, namely, a reset word serves as a \emph{constant word}
for the corresponding \emph{splicing} systems (see
e.g.~\cite{SplSyst}). Unfortunately, it is hardly believable to get
an algorithm with expected linear time for this problem, because the
problem is NP-complete. We show this result in
Section~\ref{sec_general}.

It is worth to mention that there are other types of synchronization
of partial automata, for instance \emph{careful synchronization}.
Basically, testing for synchronization becomes much more
computationally hard for these types of synchronization
(see~\cite{Mart10},\cite{Mart12} for details).

\section{Strongly Connected Partial Automata}
\label{sec_partial}

In this section we aim to adapt results
from~\cite{RandSynch},\cite{On2Problems} to the case of partial
strongly connected automata. First we should consider what we mean
by a random partial automaton. In this paper we assume that a given
transition from a state $q \in Q$ by letter $a \in \Sigma$ is
undefined equiprobable with any other possible image, that is, with
probability $\frac{1}{|Q|+1}$.

Formally, let $Q$ stand for $\{1,2, \dots n\}$, $X = \{0\} \cup Q$,
and $k>1$ for the alphabet size. Denote by ${\Sigma'}_n$ the
probability space of all maps from $X$ to $X$ which preserves $0$.
Denote by ${\Omega'}^{k}_n$ the probability space of all $k$-letter
$n$-state automata where all letters $c \in \Sigma$ are chosen
uniformly at random and independently from ${\Sigma'}_n$.

First we prove a supplementary result for the general case of
partial automaton and further get the main results as consequences.

\begin{theorem}
\label{th_main}Given a random partial automaton $\mathrsfs{A} =
\langle Q,\{a,b\},\delta \rangle \in {\Omega'}^{2}_n$, the
probability that each pair of states $\{p,q\}$ can be merged equals
$1-\Theta(\frac{1}{n})$.
\end{theorem}
\begin{proof}
Define a complete automaton $\mathrsfs{A}_c = \langle \{0\} \cup Q
,\{a,b\},\delta' \rangle$ where all undefined transition are
replaced with the transition to $0$ state.
Fix a letter $x \in \Sigma$ and remove all edges of $\mathrsfs{A}_c$
except those labeled $x$. The remaining graph is called the
\emph{underlying digraph of $x$} and is denoted $UG(x)$. Every
connected component of the underlying digraph of $x$ consists of a
unique cycle (that can degenerate to a loop) and possibly some trees
rooted on the cycle, see Fig.~\ref{fig:cluster}.
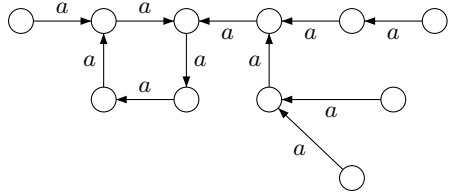
\begin{figure}[ht]
\begin{center}
\unitlength=.55mm
\begin{picture}(95,30)(0,-6)\nullfont
 \gasset{Nw=6,Nh=6,Nmr=3}
\node(B)(0,20){} \node(Y)(100,20){} \node(AD)(20,1){}
\node(AB)(20,20){} \node(BC)(40,20){} \node(CD)(40,1){}
\node(AC)(60,20){} \node(BD)(80,20){} \node(AE)(60,1){}
\node(BE)(90,1){}  \node(X)(80,-18){}
\drawedge[ELside=r](CD,AD){$a$} \drawedge(AD,AB){$a$}
\drawedge(AB,BC){$a$} \drawedge(BC,CD){$a$} \drawedge(B,AB){$a$}
\drawedge(AC,BC){$a$} \drawedge(BD,AC){$a$} \drawedge(AE,AC){$a$}
\drawedge(BE,AE){$a$} \drawedge(Y,BD){$a$} \drawedge(X,AE){$a$}
\end{picture}
\end{center}
\caption{A typical cluster}
 \label{fig:cluster}
\end{figure}
Each connected component of underlying digraphs is called
\emph{cluster}. Denote by $z_x$ the size of the cluster with $0$
($0$-cluster) and by $UG_0(x)$ the underlying digraph of $x$ without
$0$-cluster.

In what follows by \emph{wlp} we mean `with probability
$O(\frac{1}{n})$' and by \emph{whp} we mean `with probability
$1-O(\frac{1}{n})$'. The following lemma is one of the crucial ones
to adapt results to the case of partial automata.

\begin{lemma}[Appendix]
\label{lem_0cluster}Given a letter $x$ and an integer $0 \leq k \leq
n$, the probability that $z_x = k$ is at most $O( \frac{1}{k+1}(
\frac{1}{\sqrt{k+1}} + \frac{1}{\sqrt{n-k+1}}))$.
\end{lemma}

The following theorem for underlying digraphs is crucial
in~\cite{RandSynch}.
\begin{theorem}[Theorem~4 from~\cite{RandSynch}]
\label{th_high_tree}Let $g$ be the digraph of a random complete map
from $\Sigma_n$. Let $T$ be the highest tree of $g$, and denote its
height by $\tau(T)$. Then with probability $1-O(1/\sqrt{n})$ all
other trees of $g$ are lower than $\tau(T)-c$ for some constant
$c>0$ and there are at least $32\ln{n}$ vertices of levels greater
than $\tau(T)-c$ in $T$.
\end{theorem}

As an easy consequence of this theorem we get the following
corollary for the underlying digraphs of $\mathrsfs{A}$.
\begin{corollary}[Appendix]
\label{cor_htree} Let $T$ be the highest tree in $UG_0(x)$. Then
with probability $1-O(1/\sqrt{n})$ all other trees in $UG_0(x)$ are
lower than $T$ by some constant $c>0$ and there are at least
$32\ln{n}$ vertices of levels greater than $\tau(T)-c$ in $T$.
\end{corollary}

An automaton $\mathrsfs{B} = \langle Q', \Sigma, \delta' \rangle$ is
a \emph{subautomaton} of $\mathrsfs{A}$ if $Q' \subseteq Q$ and
$\delta(q,x) = \delta'(q,x)$ for each $q \in Q', x \in \Sigma$.
 The following lemma is an analogue of \cite[Lemma~1]{RandSynch}.
\begin{lemma}[Appendix]
\label{lem_reach_from_F0} The number of states in any subautomaton
of $\mathrsfs{A}$ is at least $n/4e^2$ whp.
\end{lemma}

As a straightforward consequence of Lemma~\ref{lem_reach_from_F0}
and Corollary~\ref{cor_htree} we get
\begin{corollary}
\label{cor_high_tree}Whp the underlying digraph of one letter (say
$a$) of $\mathrsfs{A}_c$ has the unique highest tree (in non
$0$-cluster) of some height $h$. Let $H$ be the set of vertices with
levels at least $h$. Then $H$ is random for letter $b$ and $H$
contains at least $32\ln{n}$ vertices.
\end{corollary}

The proof of the following lemma is almost identical to
\cite[Lemma~2]{RandSynch}.
\begin{lemma}
\label{lem_H_is_reachable}The subset $H$ from
Corollary~\ref{cor_high_tree} of top-level vertices of the
underlying digraph of $a$ intersects with any subautomaton whp.
\end{lemma}

Now let us introduce the definitions of \emph{stable} and
\emph{deadlock} pairs for partial automaton. Call a pair $\{p,q\}$
\emph{stable} if for each word $u$ such that $\{p,q\}.u$ is
non-empty there is a word $v$ such that $|\{p,q\}.uv|=1$. In
opposite, $\{p,q\}$ is called \emph{deadlock} if it can not be
merged, i.e. there is no word $u$ such that
$|\{p,q\}.u|=1$. 
A subset $A \subseteq Q$ is called $F$-clique of $\mathrsfs{A}$ if
it is a maximal by size set such that each pair of states from $A$
is deadlock. By definition all $F$-cliques have the same size.

The two following statements plays an important role in the solution
of the famous \emph{Road Coloring Problem} (see~\cite{TRRCP08}).
\begin{lemma}[Lemma~8 from~\cite{RandSynch}]
\label{lem_tr1}
    Let $A$ and $B$ be two distinct $F$-cliques
such that $A \setminus B = \{p\}, B \setminus A = \{q\}$ for some
pair of states $\{p,q\}$; Then $\{p,q\}$ is a stable pair.
\end{lemma}

\begin{theorem}[Theorem~2 from~\cite{RandSynch}]
\label{th_1stable} Suppose that the underlying digraph of $a$ has
the highest tree $T$ of height at least $t$ and all other trees are
strictly lower than $t$. Suppose also that some state $p$ of level
$t$ is reachable from $F$-clique $F_0$. Denote by $q$ the
predecessor of the root of tree $T$ on the $a$-cycle. Then
$\{p.a^{t-1},q\}$ is stable and random for $b$.
\end{theorem}
One can easily verify that the proofs of Lemma~\ref{lem_tr1} and
Theorem~\ref{th_1stable} given in~\cite{RandSynch} hold true for
partial automata also.

Given a pair $\{p,q\}$ random for a letter $x$, the probability that
$p$ or $q$ go to $0$ by $x$ equals $\frac{2}{n+1}$ while the
probability of merging by $x$ equals $\frac{1}{n+1}$. Using this
fact, one can easily verify that the following theorem
from~\cite{RandSynch} holds true for partial automata also.

\begin{theorem}[Theorem~7 from~\cite{RandSynch}]
\label{th_many_stable_ext}Whp a random $n$-state automaton
$\mathrsfs{A}=\langle Q,\{a,b\},\delta \rangle $ has $n^{0.6}$
stable pairs random for $a$ and $n^{0.6}$ stable pairs random for
$b$ and at most $O(n^{0.7})$ transitions has to be observed.
\end{theorem}
Denote by $S$ the corresponding set of stable pairs from
Theorem~\ref{th_many_stable_ext} random to letter $a$. The following
lemma gives a lower bound on the number of such pairs in $UG_0(a)$.
\begin{lemma}[Appendix]
\label{lem_stable_in_ac}Let $k = z_a$. If $n-k \geq n^{0.701}$ then
there are at least $(n-k)^{0.001}$ pairs from $S$ in $UG_0(a)$ whp.
\end{lemma}

Call a set of states a \emph{synchronizing class} if each pair from
this set can be merged. Due to Lemma~\ref{lem_stable_in_ac} we can
adapt \cite[Corollary~2]{RandSynch} to the underlying digraphs for
both letters.
\begin{corollary}
\label{cor_big_syn_clusters}If $n-z_x \geq n^{0.701}$ for $x \in
\{a,b\}$. Then whp there is at most $5\ln{n}$ clusters of $UG_0(x)$
in a one synchronizing class of common size at least $(n-z_x) -
(n-z_x)^{0.45}$.
\end{corollary}

Let $x \in \{a,b\}$ and $n-z_x \geq n^{0.701}$.
Corollary~\ref{cor_big_syn_clusters} implies that all clusters
greater than $n^{0.45}$ of $UG_0(x)$ lie in a one synchronizing
class. Denote the set of states in these clusters for $x$ by $S_x$.
Notice that this set is random for the second letter because it is
completely defined by $x$. Denote by $T_x$ the complement for $S_x$
in $UG_0(x)$. Equivalently, $T_x$ can be defined as the state set of
the clusters of $UG_0(x)$ of size at most $n^{0.45}$. Since whp
there are at most $5\ln{n}$ clusters, we also get that $|T_x| \leq
5\ln{(n)}n^{0.45} \leq n^{0.46}$ whp for $n$ big enough.

The following lemma gives an upper bound on the number of states
with undefined transitions for each letter.
\begin{lemma}[Appendix]
\label{lem_max_und_trans} Whp there are at most $\ln{n}$ states with
undefined transition by $x$ for each $x \in \{a,b\}$.
\end{lemma}

\begin{lemma}[Appendix]
\label{lem_rand_pair} A random pair $\{p,q\}$ for a letter $x \in
\{a,b\}$ is deadlock with probability at most
$O(\frac{1}{n^{0.51}})$.
\end{lemma}

The following corollary easily follows from
Lemma~\ref{lem_rand_pair}.
\begin{corollary}
\label{cor_rand_pair} A random pair $\{p,q\}$ for a letter $x \in
\{a,b\}$ is deadlock with probability $O(\frac{1}{n^{1.02}})$.
\end{corollary}
\begin{proof}
    Without loss of generality, suppose $x=a$. Since $\{p,q\}$ is
random for $a$, the sets $\{p.a,q.a\}, \{p.a^2,q.a^2\}$ are non
empty with probability at least $1-\frac{2}{(n+1)^2}$ and random for
$b$. If $|\{p.a,q.a\}|=1$ or $|\{p.a^2,q.a^2\}|=1$ the pair
$\{p,q\}$ is not deadlock. Otherwise by Remark~\ref{lem_rand_pair}
one of these pairs is not deadlock with probability
$1-O(\frac{1}{n^{1.02}})$ whence $\{p,q\}$ also.
\end{proof}

The following lemma completes the proof of the lower bound.
\begin{lemma}
\label{lem_deadlock} $\mathrsfs{A}$ does not have deadlock pairs
whp.
\end{lemma}
\begin{proof}
Suppose there is a deadlock pair $\{p,q\}$. Consider first the case
when $0$-cluster of some letter (say $a$) is reachable from
$\{p,q\}$. Then there is a deadlock pair $\{p',q'\}$ such that both
$p'.a$ and $q'.a$ are undefined. By Lemma~\ref{lem_max_und_trans}
there are at most $\ln^2{n}$ of such pairs whp. Notice that these
pairs are random for $b$. By Corollary~\ref{cor_rand_pair} one of
these pairs is deadlock with probability at most $\ln^2{n}
\frac{1}{n^{1.02}} = O(\frac{1}{n})$.

Now consider the case when $0$-clusters are not reachable from
$\{p,q\}$. This means that there is a complete subautomaton
$\mathrsfs{B}$ of $\mathrsfs{A}$ reachable from $\{p,q\}$. By
Lemma~\ref{lem_reach_from_F0} the size of $\mathrsfs{B}$ is at least
$n/4e^2$. Clearly $\mathrsfs{B}$ is a random complete automaton of
size $n/4e^2$ whence by Theorem~1 from~\cite{RandSynch} it is
synchronizable whp whence such a pair $\{p,q\}$ exists wlp as
required.
\end{proof}

Notice that $\mathrsfs{A}$ is complete with probability
$(\frac{n}{n+1})^{2n} \geq 0.5e^2$ for $n$ big enough. Now the lower
bound $1-\Theta(\frac{1}{n})$ follows from
\cite[Theorem~1]{RandSynch}.
\end{proof}

Since each strongly connected subautomaton of a random $n$-state
automaton is also random, and by Lemma~\ref{lem_reach_from_F0} its
size is at least $n/4e^2$ whp, we get the main result as a
straightforward consequence of Criterion~\ref{e_crit} and
Theorem~\ref{th_main}.
\begin{theorem}
\label{th_main}The probability of being synchronizable for
$2$-letter strongly connected partial random automaton with $n$
states is $1-O(\frac{1}{n})$.
\end{theorem}

The following theorem is an analogue
of~\cite[Theorem~1]{On2Problems}.
\begin{theorem}
\label{th_opt_alg} There is a deterministic algorithm that verifies
Criterion~\ref{e_crit} for a given $k$-letter partial strongly
connected automaton. The proposed algorithm works in linear expected
time in $n$ with respect to ${\Omega'}^{k}_n$. Moreover, for this
problem the proposed algorithm is optimal by expected time up to a
constant factor.
\end{theorem}
\begin{proof}
The only difference with the algorithm presented
in~\cite[Theorem~1]{On2Problems} for complete automata is that here
we are based on Theorem~\ref{th_main}. One can easily verify that
all additional routines for partial automata can be also done in
linear time.
\end{proof}

Finalizing this section, let us remark that all the results can be
trivially adapted to any fixed non-singleton alphabet.

\section{Reachable Partial Automata}
\label{sec_general}

In opposite to the complete automata case, the general case of
partial automata can not be polynomially reduced to the strongly
connected case. Namely, in this section we prove that the
synchronization testing for reachable partial automata is
NP-complete problem.

Clearly for unary alphabet case the problem has no sense because the
unique synchronizing automaton should be complete. Thus $2$-letter
alphabet case is the most interesting. The proof is similar to one
given in~\cite{MyTOCS2013} for the problem of approximating the
minimum length of synchronizing words. We first present a
construction for $4$-letter alphabet and further transform it into
$2$-letter alphabet case using standard encoding techniques.

\begin{theorem}
\label{thm_4} Testing a given reachable partial $4$-letter automaton
for synchronization is NP-complete problem.
\end{theorem}
\begin{proof}
If a given partial $n$-state automaton is synchronizing then it has
a reset word of length at most $n^3/2$. This bound easily follows
from Criterion~\ref{e_crit} because the pairs can be merged
subsequently, and if a pair can be merged then it can be merged by a
word of length at most $n^2/2$. Hence there is a polynomial size
certificate for a given instance whence the problem belongs to NP.

Let us take an arbitrary instance $\psi$ of the classical
NP-complete problem SAT (the satisfiability problem for a system of
clauses, that is, formulae in conjunctive normal form) with $n$
variables $x_1, x_2, \dots ,x_n$ and $m$ clauses $c_1, c_2, \dots
,c_m$. For convenience we may assume that the number of clauses $m$
coincides with $n+1$. Otherwise we can either add $m-(n+1)$ fake
variables if $m<n+1$ or $n+1-m$ clauses $(x_1 \cup \neg x_1)$ if
$m>n+1$. We shall construct a reachable partial automaton
$\mathrsfs{A}(\psi)$ with $4$ input letters and polynomial in $m,n$
number of states such that $\mathrsfs{A}$ is synchronizing if and
only if $\psi$ is satisfiable.

Now we describe the construction of the automaton
$\mathrsfs{A}(\psi)=\langle Q,\Sigma,\delta\rangle$ where
$\Sigma=\{a,b,c,d\}$. The state set $Q$ of $\mathrsfs{A}(\psi)$ is
the disjoint union of the three following sets:
\begin{align*}
S^q&= \{q_{i,j} \mid 1 \leq i \leq m,\, 1 \leq j \leq n+1\},\\
S^p&= \{p_{i,j} \mid 1 \leq i \leq m,\, 1 \leq j \leq i \},\\
S^g&= \{g_{i,j} \mid 0 \leq i \leq 1,\, 1 \leq j \leq n+1\}.
\end{align*}
The size of $Q$ is equal to $m(n+1)+m(m+1)/2+2(n+1)$, and hence is a
polynomial in $m,n$.

Now the transition function $\delta$ is defined as follows:
$$\delta(q,a) =
\begin{cases}
    p_{i,j+1} & \text{if } q = p_{i,j} \text{ and } j \leq n;\\
    g_{i,j+1} & \text{if } q = g_{i,j} \text{ and } j \leq n;\\
    q_{i,j+1} & \text{if } q = q_{i,j} \text{ and } x_j \not\in c_i.\\
\end{cases}$$
$$\delta(q,b) =
\begin{cases}
    p_{i,j+1} & \text{if } q = p_{i,j} \text{ and } i < j \leq n;\\
    p_{i+1,i+1} & \text{if } q = p_{i,i} \text{ and } i \leq n;\\
    g_{i,j+1} & \text{if } q = g_{i,j} \text{ and } j \leq n;\\
    q_{i,j+1} & \text{if } q = q_{i,j} \text{ and } \neg x_j \not\in c_i.\\
\end{cases}$$
$$\delta(q,c) =
\begin{cases}
    g_{0,4} & \text{if } q \in S^p;\\
    g_{1,4} & \text{if } q \in S^q;\\
    g_{i,j+1} & \text{if } q = g_{i,j} \text{ and } j \leq n.\\
\end{cases}$$
$$\delta(q,d) =
\begin{cases}
    q_{i,1} & \text{if } q = p_{i,n+1};\\
    g_{1-i,1} & \text{if } q = g_{i,1};\\
    g_{i,1} & \text{if } q = g_{i,j} \text{ and } j > 1.\\
\end{cases}$$

Let us informally comment on the essence of our construction. It is
based on Eppstein's gadget $\mathrsfs{E}(\psi)$ from~\cite{Ep90}.
The gadget consists of the state set $\{q_{i,j} \mid 1 \leq i,j \leq
n+1\}$, on which the letters $a$ and $b$ act as described above, and
controls the following. If the literal $x_j$ (respectively $\neg
x_j$) occurs in the clause $c_i$, then the letter $a$ (respectively
$b$) is undefined on the state $q_{i,j}$. This encodes the situation
when one can satisfy the clause $c_i$ by choosing the value $1$
(respectively $0$) for the variable $x_j$. Otherwise, the letter $a$
(respectively $b$) increases the second index of the state. This
means that one cannot make $c_i$ be true by letting $x_j=1$
(respectively $x_j=0$), and the next variable has to be inspected.

For the reader's convenience, we illustrate the construction of
$\mathrsfs{A}(\psi)$ on the following example.
Figure~\ref{A2_example} shows two automata of the form
$\mathrsfs{A}(\psi)$ built for the SAT instances
\begin{align*}
\psi_1&=\{x_1 \vee x_2 \vee x_3,\, \neg x_1 \vee x_2,\, \neg x_2
\vee x_3,\,\neg x_2 \vee \neg x_3\},\\
\psi_2&=\{x_1 \vee x_2,\,\neg x_1 \vee x_2,\, \neg x_2 \vee
x_3,\,\neg x_2 \vee \neg x_3\}.
\end{align*}

\unitlength=0.85mm

\begin{figure}[ht]
\scalebox{0.7}[0.7]{
\begin{picture}(120,170)(0,-175)

\drawrect[dash={5.0 3.0}{0.0}](17.72,-17.705,120,-78.295)
\drawrect[dash={5.0 3.0}{0.0}](17.72,-83.295,100,-168.295)
\node[Nframe=n,Nw=10.32,Nh=9.0,Nmr=0.0](Sq)(90,-165.0){$S^q$}

\node[Nframe=n,Nw=10.32,Nh=9.0,Nmr=0.0](Sp)(115,-72.0){$S^p$}

\node(p11)(32.24,-24.11){$p_{1,1}$}
\node(p12)(32.30,-40.08){$p_{1,2}$}
\node(p13)(32.30,-56.05){$p_{1,3}$}
\node(p14)(32.32,-72.02){$p_{1,4}$}

\node(p22)(52.29,-40.08){$p_{2,2}$}
\node(p23)(52.29,-56.05){$p_{2,3}$}
\node(p24)(52.29,-72.02){$p_{2,4}$}

\node(p34)(72.26,-72.02){$p_{3,4}$}
\node(p33)(72.26,-56.05){$p_{3,3}$}

\node(p44)(92.07,-71.94){$p_{4,4}$}

\drawedge(p11,p12){$a$} \drawedge(p11,p22){$b$}
\drawedge(p12,p13){$a,b$} \drawedge(p13,p14){$a,b$}

\drawedge(p22,p33){$b$} \drawedge(p22,p23){$a$}
\drawedge(p23,p24){$a,b$}

\drawedge(p33,p44){$b$} \drawedge(p33,p34){$a$}

\node(q32)(71.98,-111.80){$q_{3,2}$}
\node(q22)(52.01,-111.80){$q_{2,2}$}
\node(q23)(52.01,-131.74){$q_{2,3}$}
\node(q24)(52.01,-151.68){$q_{2,4}$}
\node(q34)(71.98,-151.68){$q_{3,4}$}
\node(q21)(51.73,-91.68){$q_{2,1}$}
\node(q31)(71.92,-91.86){$q_{3,1}$}
\node(q33)(71.98,-131.74){$q_{3,3}$}
\node(q42)(91.95,-111.80){$q_{4,2}$}
\node(q44)(92.09,-151.68){$q_{4,4}$}
\node(q41)(91.94,-91.91){$q_{4,1}$}
\node(q43)(92.09,-131.74){$q_{4,3}$}
\node(q12)(32.02,-111.80){$q_{1,2}$}
\node(q14)(32.04,-151.68){$q_{1,4}$}
\node(q11)(31.96,-91.86){$q_{1,1}$}
\node(q13)(32.02,-131.74){$q_{1,3}$}

\node(g01)(120.00,-91.86){$g_{0,1}$}
\node(g02)(120.00,-112.00){$g_{0,2}$}
\node(g03)(120.00,-131.74){$g_{0,3}$}
\node(g04)(120.00,-152.54){$g_{0,4}$}

\drawedge(g01,g02){$a,b,c$} \drawedge(g02,g03){$a,b,c$}
\drawedge(g03,g04){$a,b,c$}

\node(g11)(145.00,-91.86){$g_{1,1}$}
\node(g12)(145.00,-112.00){$g_{1,2}$}
\node(g13)(145.00,-131.74){$g_{1,3}$}
\node(g14)(145.00,-152.54){$g_{1,4}$}

\drawedge(g11,g12){$a,b,c$}
\drawedge(g12,g13){$a,b,c$} \drawedge(g13,g14){$a,b,c$}

\drawedge(g01,g11){$d$}
\drawedge[curvedepth=5.0,ELside=r,ELpos=40](g02,g01){$d$}
\drawedge[curvedepth=8.0,ELside=r,ELpos=30](g03,g01){$d$}
\drawedge[curvedepth=11.0,ELside=r,ELpos=30](g04,g01){$d$}

\drawedge(g11,g01){}
\drawedge[curvedepth=5.0,ELside=r,ELpos=40](g12,g11){$d$}
\drawedge[curvedepth=9.0,ELside=r,ELpos=30](g13,g11){$d$}
\drawedge[curvedepth=13.0,ELside=r,ELpos=30](g14,g11){$d$}

\drawedge(p14,q11){$d$} \drawedge(p24,q21){$d$}
\drawedge(p34,q31){$d$} \drawedge(p44,q41){$d$}

\drawedge(q11,q12){$b$} \drawedge(q21,q22){$a$}
\drawedge(q31,q32){$a,b$} \drawedge(q41,q42){$a,b$}

\drawedge(q12,q13){$b$} \drawedge(q22,q23){$b$}
\drawedge(q32,q33){$a$} \drawedge(q42,q43){$a$}

\drawedge[dash={3.0
3.0}{0.0},ELside=r,ELdist=2.0,curvedepth=-8.6](q13,q14) {$a$ in
$\mathrsfs{A}(\psi_2)$} \drawedge[curvedepth=8.0](q13,q14){$b$}

\drawedge(q23,q24){$a,b$} \drawedge(q33,q34){$b$}
\drawedge(q43,q44){$a$}

\drawedge[curvedepth=-14.0,ELpos=25](Sp,g04){$c$}
\drawedge[ELside=r,ELpos=50](Sq,g14){$c$}

\end{picture}
}
\caption{The automata $\mathrsfs{A}(\psi_1)$ and
$\mathrsfs{A}(\psi_2)$} \label{A2_example}

\end{figure}
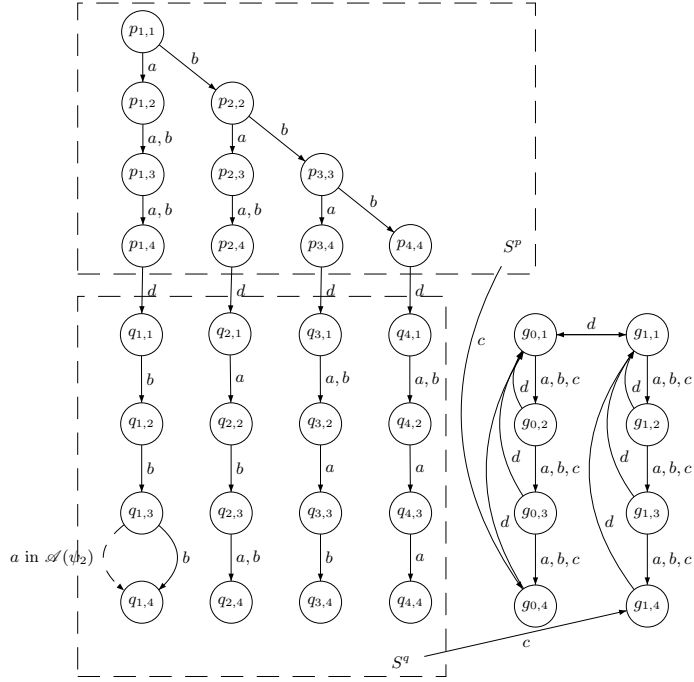

The two instances differ only in the first clause: in $\psi_1$ it
contains the variable $x_3$ while in $\psi_2$ it does not.
Correspondingly, the automata $\mathrsfs{A}(\psi_1)$ and
$\mathrsfs{A}(\psi_2)$ differ only by the outgoing arrow labeled $a$
at the state $q_{1,3}$: in $\mathrsfs{A}(\psi_1)$ there is no such
arrow while in $\mathrsfs{A}(\psi_2)$ it leads to the state
$q_{1,4}$ and is shown by the dashed line.

Observe that $\psi_1$ is satisfiable for the truth assignment
$x_1=0$, $x_2=0$, $x_3=1$ while $\psi_2$ is not satisfiable. It is
not hard to check that the word $bbac$ synchronizes
$\mathrsfs{A}(\psi_1)$ to the state $g_{0,4}$ and
$\mathrsfs{A}(\psi_2)$ is not synchronizing.

We may assume that $\psi$ is \emph{reduced}, i.e. for each $j>1$ at
most one of the literals $x_j, \neg x_j$ may belong to some clause
$c_i$. This would imply that $\mathrsfs{A}(\psi)$ is reachable.

First consider the case when $\psi$ is satisfiable. Then there
exists a truth assignment
$$\tau:\{x_1,\dots,x_n\}\to\{0,1\}$$ such that
$c_i(\tau(x_1),\dots,\tau(x_n))=1$ for every clause $c_i$ of $\psi$.
We construct a word $v=v(\tau)$ of length $n$ as follows:
\begin{equation}
\label{encoding} v[j]=\begin{cases}
a & \text{ if } \tau(x_j)=1;\\
b & \text{ if } \tau(x_j)=0.
\end{cases}
\end{equation}
We aim to prove that the word $w=vc$ is a \sw\ for
$\mathrsfs{A}(\psi)$, that is, $|Q.w|=1$. Clearly, $q_{i,j}.v =
q_{i,j+n}.x$ is undefined for $j>1$ because $x \in \{a,b\}$.
Analogously, $S^g.w = \emptyset$. Since
$$c_i(\tau(x_1),\dots,\tau(x_n))=1$$ for every clause $c_i$, there
is an index $j$ such that either $x_j \in c_i$ and $\tau(x_j)=1$ or
$\neg x_j \in c_i$ and $\tau(x_j)=0$. This readily implies (see the
comment following the definition of the transition function of
$\mathrsfs{A}(\psi)$) that $q_{i,1}.v$ is undefined for all $1\le
i\le m$. On the other hand, $S^p.w = p_{1,1}.vc = g_{0,n+1}$ because
there is exactly one valid path from $S^p$ by word of length $n$
that does not involve $c$. Thus we have shown that $w$ is reset for
$\mathrsfs{A}(\psi)$.

Now we consider the case when $\psi$ is not satisfiable. Arguing by
contradiction, let $w$ be the shortest \sw. The following remark
easily follows from the definition of $\delta$ on $S^g$ and $S^q$.
\begin{remark}
\label{rem_bad_g} If a set $T \subseteq S^g \cup S^q$ contains a
pair $\{g_{0,j_{min}},g_{1,j_{min}}\}$ where $j_{min}$ is the
minimum row index of states from $T$; then $T$ can not be merged.
\end{remark}

Suppose $i \leq n$ be the first position of $c$ or $d$ in $w$. If
$w[i]=c$ then $Q.w[1..i] \subseteq S^g$ and $Q.w[1..i]$ satisfies
Remark~\ref{rem_bad_g} because $Q.w[1..i-1]$ contains states from
both $S^p$ and $S^q$ and $S^p.c = \{g_{0,n+1}\}, S^q.c =
\{g_{1,n+1}\}$. If $w[i]=d$ then $\{g_{0,1},g_{1,1}\} \subseteq
Q.w[1..i]$ and $S^p \cap Q.w[1..i] = \emptyset$ whence $Q.w[1..i]$
satisfies Remark~\ref{rem_bad_g} again.

Thus $w = u v$ where $u \in \{a,b\}^n$. Define a truth assignment
$\tau:\{x_1,\dots,x_n\}\to\{0,1\}$ as follows:
$$\tau(x_j)=\begin{cases}
1 & \text{ if } u[j]=a;\\
0 & \text{ if } u[j]=b.
\end{cases}$$
Since $\psi$ is not satisfiable, we have
$c_i(\tau(x_1),\dots,\tau(x_n))=0$ for some clause $c_i$, $1\le i\le
m$. According to our definition of the transition function of
$\mathrsfs{A}(\psi)$, this means that $q_{i,j}.u[j]=q_{i,j+1}$ for
all $j=1,\dots,n$. Hence $q_{i,n+1}= q_{i,1}.w[1..n]$.

If $w[n+1] \in \{a,b\}$ then $Q.w = \emptyset$. If $w[n+1] = c$ then
$Q.w[1..n+1] = \{g_{0,1},g_{1,1}\}$ and by Remark~\ref{rem_bad_g} we
get a contradiction. Finally if $w[n+1]=d$ then $\{g_{0,1},g_{1,1}\}
\subseteq Q.w[1..i]$ and $S^p \cap Q.w[1..i] = \emptyset$ whence
$Q.w[1..i]$ again satisfies Remark~\ref{rem_bad_g}. Thus we get a
contradiction whence $\mathrsfs{A}(\psi)$ is not synchronizing and
we are done. \qed
\end{proof}

Now we show that Theorem~\ref{thm_4} can be extended to automata
with only 2 input letters.
\begin{corollary}
\label{cor_2} The problem of deciding whether a given reachable
partial $2$-letter automaton is synchronizing is NP-complete.
\end{corollary}
\begin{proof}
For every partial automaton
$\mathrsfs{A}=(Q,\Sigma=\{a_1,a_2,a_3,a_4\},\delta)$, we construct a
reachable automaton $\mathrsfs{B}=(Q',\{a,b\}, \delta')$ such that
$\mathrsfs{A}$ is synchronizing if and only if $\mathrsfs{B}$ is
synchronizing and $|Q'|$ is a polynomial of $|Q|$. We let $Q' = Q \
\times \Sigma$ and define the function $\delta':Q'\times \{a,b\}\to
Q'$ as follows:
\begin{align*}
\delta'((q,a_i),a)&=(q,a_{\min(i+1,4)}),\\
\delta'((q,a_i),b)&=(\delta(q,a_i),a_1).
\end{align*}
Thus, the action of $a$ on a state $q'\in Q'$ substitutes an
appropriate letter from the alphabet $\Sigma$ of $\mathrsfs{A}$ for
the second component of $q'$ while the action of $b$ imitates the
action of the second component of $q'$ on its first component and
resets the second component to $a_1$. Given a word $w = a^{i_1} b
a^{i_2} b \dots a^{i_k} b \in \Sigma'^{*}$ define a word $f(w) =
a_{\min(i_1,4)}a_{\min(i_2,4)} \dots a_{\min(i_k,4)}$.

Then the word $f(w)$ is easily seen to be a \sw\ for $\mathrsfs{A}$
if and only if $w$ is \sw\ for $\mathrsfs{B}$. The corollary follows
because $f$ is bijective function from $\Sigma'^{*}$ to
$\Sigma^{*}$.\qed
\end{proof}

\newpage
\section*{Appendix}

\textbf{Lemma~\ref{lem_0cluster}.} Given a letter $x$ and an integer
$0 \leq k \leq n$, the probability that $z_x = k$ is at most $O(
\frac{1}{k+1}( \frac{1}{\sqrt{k+1}} + \frac{1}{\sqrt{n-k+1}}))$.
\begin{proof}
    We use the famous formula $N(n+N)^{n-1}$ for the number of
forests with $N$ root vertices and $n$ non-root vertices. Then
\begin{equation}
\label{eq_zx} P(z_x = k) = \frac{{n \choose k} (k+1)^{k-1}
(n-k)^{n-k}}{(n+1)^n}.
\end{equation} Indeed, first we choose $k$ subset from $Q$ in ${n
\choose k}$ ways, next we choose a tree with root $0$ and $k$
non-root vertices in $(k+1)^{k-1}$ ways; the transitions by $x$ for
remaining $n-k$ states can be defined in $(n-k)^{n-k}$ ways. Since
there are $(n+1)^n$ ways to choose $x \in {\Sigma'}_n$,
equation~(\ref{eq_zx}) follows.

Using Stirling's formula we get that $${n \choose k} (n-k)^{n-k}
\leq \frac{n^n}{k^k} \frac{n}{\sqrt{(k+1)(n-k+1)}}.$$ Finally we get
the bound
$$\frac{\sqrt{n}}{\sqrt{(k+1)(n-k+1)}(k+1)} \frac{(1+1/k)^{k}}{(1+1/n)^{n}} \leq c \frac{1}{k+1}( \frac{1}{\sqrt{k+1}} + \frac{1}{\sqrt{n-k+1}}).$$
\end{proof}

\textbf{Corollary~\ref{cor_htree}.} Let $T$ be the highest tree in
$UG_0(x)$. Then with probability $1-O(1/\sqrt{n})$ all other trees
in $UG_0(x)$ are strictly lower than $T$ and there are at least
$32\ln{n}$ vertices of levels greater than $\tau(T)$ in $T$.
\begin{proof}
    For a given $k=z_x$ the digraph $UG_0(x)$ is a random digraph of size $n-z_x$.
Hence the probability that this digraph does not satisfy
Theorem~\ref{th_high_tree} can be bounded by
\begin{multline}
2c \sum_{k=0}^{n} \frac{1}{\sqrt{n-k+1}} \frac{1}{k+1}(
\frac{1}{\sqrt{k+1}} + \frac{1}{\sqrt{n-k+1}}) = \\
= 2c \sum_{k=1}^{n+1} (\frac{1}{(n+2-k)k} + \frac{1}{k^{1.5}
\sqrt{n+2-k}}) \leq \\ \leq 2c( \frac{\ln{n}}{n} +
\sum_{k=1}^{n/2}{\frac{1}{k^{1.5} \sqrt{n/2}}} +
\sum_{k=n/2}^{n+1}{\frac{1}{(n/2)^{1.5} \sqrt{n+2-k}}}) \leq \\
\leq 2c( o(\frac{1}{\sqrt{n}}) + \int_{1}^{n/2}{\frac{dk}{k^{1.5}
\sqrt{n/2}}} + \int_{n/2}^{n+1}{\frac{dk}{(n/2)^{1.5}
\sqrt{n+2-k}}}) = \\
 = o(\frac{1}{\sqrt{n}}) + 0.5c
(\frac{1}{\sqrt{k} \sqrt{n/2}} |_{n/2}^{1} + \frac{1}{(n/2)^{1.5}}
\sqrt{n+2-k} |_{n}^{n/2}) = O(\frac{1}{\sqrt{n}}).
\end{multline}
\end{proof}

\textbf{Lemma~\ref{lem_reach_from_F0}.}
The number of states in any
subautomaton of $\mathrsfs{A}$ is at least $n/4e^2$ whp.
\begin{proof}
The probability that there is subset of size less than $n/4e^2$
which is closed under the actions of the letters can be bounded by
$$\sum_{i=1}^{n/4e^2}{n \choose i}(\frac{i}{n})^{2i} \leq \sum_{i=1}^{n/4e^2}(\frac{ei}{n})^{i}(\frac{(i+1)n}{i(n+1)})^{2i} \leq e^2 \sum_{i=1}^{n/4e^2}(\frac{ei}{n})^{i}.$$
Indeed, we first choose an $i$-state subset in ${n \choose i}$ ways
and then the probability that both letters leave one state in this
subset is $(\frac{i+1}{n+1})^{2}$.

For $i \leq n/4e^2$ we get that $(\frac{ei}{n})^{i} \geq
2(\frac{e(i+1)}{n})^{i+1}$. Hence the sum can be bounded by doubled
first element $2e^2(\frac{e}{n})$ and we are done.
\end{proof}

\textbf{Lemma~\ref{lem_stable_in_ac}.} Let $k = z_a$. If $n-k \geq
n^{0.701}$, then there are at least $(n-k)^{0.001}$ pairs from $S$
in $UG_0(a)$ whp.
\begin{proof}
Since $S$ is random for $a$, we have to estimate the probability of
choosing $|S|$ random distinct pairs without repetition from $Q
\times Q$ such that we will choose less than $(n-k)^{0.001}$ pairs
in $UG_0(a)$. Suppose we have already chosen $d < |S|$ pairs, and at
most $n^{0.001}$ of these pairs lie inside $S$. Then the probability
to choose the next pair inside $S$ is at least
$$ \frac{(n-k-2d)^2}{n^2} \geq \frac{(1-\frac{2d}{n-k})^2
(n-k)^2}{n^2} \geq (1-\frac{2n^{0.6}}{n^{0.701}})^2
\frac{(n-k)^2}{n^2}.$$ Since $(1-\frac{2n^{0.6}}{n^{0.701}})^{2
n^{0.001}} = 1-O(\frac{1}{n^{0.1}}) \geq 2$ for sufficiently big
$n$, the factor $p=(1-\frac{2n^{0.6}}{n^{0.701}})^2$ does not impact
on the asymptotic. Hence we can consider the choice of $S$ as the
Bernoulli scheme with $n^{0.6}$ independent experiments, each of
which yields success with probability $p$. Then the number of
successes is given by binomial distribution. Using Chernoff's
inequality $F(n,k,p) \leq e^{-\frac{1}{2p}\frac{(np-k)^2}{n}}$ for
the cumulative distribution function of the binomial distribution we
get the desired bound
\begin{multline}F(|S|,(n-k)^{0.001},\frac{(n-k)^2}{n^2}) \leq
e^{-\frac{(n^{0.6}\frac{(n-k)^2}{n^2} - (n-k)^{0.001})^2}{2
\frac{(n-k)^2}{n^2} n^{0.6}}} \leq \\
\leq e^{-\frac{(n^{0.6}\frac{(n-k)^2}{n^2})}{4}} = e^{-0.002n} =
o(\frac{1}{n}).
\end{multline}

\end{proof}

\textbf{Lemma~\ref{lem_max_und_trans}.} Whp there are at most
$\ln{n}$ states with undefined transition by $x$ for each $x \in
\{a,b\}$.
\begin{proof}
    Suppose there are exactly $r$ states with undefined transition
by $x$. The probability of such event is $$\frac{{n \choose
r}n^{n-r}}{(n+1)^n} \leq \frac{n^n n^{n-r}}{r^r (n-r)^{n-r} (n+1)^n}
\leq (\frac{e}{r})^r = \phi(r)$$ Indeed, there are ${n \choose r}$
ways to choose $r$ states with undefined transition by $x$, and for
each of the $n-r$ remained states there are $n$ ways to define
transition by $x$. Since $\phi(r)>2\phi(r+1)$ for $r>2e$, the
probability of more than $\ln{n}$ undefined transitions (for
$n>e^{2e}$) is bounded by $\phi(\ln{n}) = e^{(1-\ln{\ln{n}})\ln{n}}$
which is $o(\frac{1}{n})$.
\end{proof}

\textbf{Lemma~\ref{lem_rand_pair}.} A random pair $\{p,q\}$ for a
letter $x \in \{a,b\}$ is deadlock with probability at most
$O(\frac{1}{n^{0.51}})$.
\begin{proof}
Without loss of generality suppose $x=a$ and consider the chain of
states $$p,q, p.a,q.a, \dots ,p.a^{r-1},q.a^{r-1}$$ where $r$ is the
maximal integer such that all states in this chain are different and
defined.

Suppose $p' = p.a^{r}$ is defined and already exists in the chain.
If $q'.a^{r+t}$ is not defined for some $t \geq 0$, then the pair
$\{p,q\}$ is not deadlock, because $p.a^{r (t+1)} = p'$ is defined
while $q.a^{r (t+1)}$ does not. Hence $\{p,q\}$ belongs to
$UG_0(a)$. If $n-z_a<n^{0.701}$ the probability of such event is at
most $2(\frac{n-z_a}{n})^2 = O({n^{-0.299*2}}) =
O(\frac{1}{n^{0.51}})$. Otherwise at least one of the states must
belong to $T_{a}$, and by Corollary~\ref{cor_big_syn_clusters} the
probability of such event is at most $2\frac{|T_a|}{n} =
O(\frac{1}{n^{0.51}})$.

It remains to consider the case when both $p.a^{r}$ and $q.a^{r}$
are not defined. The probability of such event is at most
$\frac{1}{(n+1)^2}$ for a given pair and for the whole chain is
bounded by
$$\frac{1}{(n+1)^2}(1 + \frac{(n-2)(n-3)}{(n+1)^2} +
\frac{(n-4)(n-5)}{(n+1)^2} + \dots +
\frac{(n-2r)(n-2r-1)}{(n+1)^2}).$$ Clearly $r < 0.5n$ whence the sum
is bounded by $$\frac{2}{(n+1)^4}\int_{0}^{0.5n}{(n-2r)^2 dr} =
O(\frac{1}{n}).$$
\end{proof}

\end{document}